\documentclass[a4paper,10pt]{article}
\usepackage{amsmath}
\usepackage{amssymb}
\usepackage{amsthm}
\usepackage{tikz}

\newtheorem{thm}{Theorem}[section]
\newtheorem{prop}[thm]{Proposition}
\newtheorem{lemma}[thm]{Lemma}

\newtheorem{definition}[thm]{Definition}
\theoremstyle{definition}
\newtheorem{remark}[thm]{Remark}
\newcommand{\id}{\mathrm{id}}

\begin{document}

\begin{center}
{\huge{On the lattice-geometry and birational group of the six-point multi-ratio equation}}
\vskip10mm
{\large James Atkinson}\\
\vskip3mm
{\small
Northumbria University, Newcastle upon Tyne, UK, and the University of Sydney, New South Wales, Australia.
}
\end{center}
\noindent {\bf Abstract.}
The inherent self-consistency properties of the six-point multi-ratio equation allow it to be considered on a domain associated with a T-shaped Coxeter-Dynkin diagram.
This extends the KP lattice, which has $A_N$ symmetry, and incorporates also KdV-type dynamics on a sub-domain with $D_N$ symmetry, and Painlev\'e dynamics on a sub-domain with $\tilde{E}_8$ symmetry.
More generally, it can be seen as a distinguished representation of Coble's Cremona group associated with invariants of point sets in projective space.

\section{Introduction}\label{INTRO}
The six-point multi-ratio equation,
\begin{equation}
\frac{(x_1-\bar{x}_2)(x_2-\bar{x}_3)(x_3-\bar{x}_1)}{(x_1-\bar{x}_3)(x_3-\bar{x}_2)(x_2-\bar{x}_1)}=1,\label{skp}
\end{equation}
is best known in the theory of integrable systems as a discrete analogue of the Kadomtsev-Petviashvili (KP) equation.
The multi-ratio form emerges naturally as the superposition principle of the Schwarzian KP equation \cite{DN}, and in the incidence-geometry approach to the KP hierarchy \cite{slg}.
It is directly related to earlier forms \cite{HirotaKP,NCWQ}.

A lattice with $A_N$ symmetry has played a central role in several contemporary developments in the theory of the discrete KP systems \cite{DDM,DAWG,ABSo,ks2}.
One reason to consider the $A_N$ lattice as the domain for the multi-ratio form (\ref{skp}), is based purely on symmetry arguments, and the purpose of this paper is to present an extension which is natural in the same sense.
The extended domain for (\ref{skp}) is encoded in the Coxeter-graph of Figure \ref{cdd}.
The situation is that sub-groups correspond to sub-domains, and for instance the $A_N$ KP lattice corresponds to the case $k=0$.
Truncations of the graph satisfying the condition \cite{rp}
\begin{equation}\label{cc}
i+j+k+1\ge ijk-1,
\end{equation}
imply the associated group is finite or affine, corresponding to inequality or equality respectively.
Such truncation is the simplest way to pick out sub-domains which support integrable dynamics in a familiar sense.

Besides the KP case, two other sub-cases of the considered domain for (\ref{skp}) are of particular significance.
In the case $k=j=1$, the system is equivalent to the discrete Schwarzian Korteweg-de Vries (KdV) system \cite{NC} in multi-dimensions \cite{NW,BS1}.
In the case $k=1$, $j=2$, $i=6$, it gives a representation of the elliptic Painlev\'e system \cite{sc,org}.
\begin{figure}[t]
\begin{center}
\begin{tikzpicture}[thick]
  \node at (4,4) [draw,circle,fill=white,minimum size=10pt,inner sep=0pt]{};
  \tikzstyle{every node}=[draw,circle,fill=black,minimum size=5pt,inner sep=0pt];
  \node (r1) at (4,2) [label={[label distance=2pt]right:$1$}]{};
  \node (r3) at (4,4) [label={[label distance=2pt]right:$k$}]{};
  \node (s1) at (5,1) [label={[label distance=2pt]below:$1$}]{};
  \node (s3) at (7,1) [label={[label distance=2pt]below:$i$}]{};
  \node (t1) at (3,1) [label={[label distance=2pt]below:$1$}]{};
  \node (t3) at (1,1) [label={[label distance=2pt]below:$j$}]{};
  \node (sig) at (4,1) [label={[label distance=2pt]below:$0$}]{};
  \draw (sig)--(r1)--(4,2.3);
  \draw (sig)--(s1)--(5.3,1);
  \draw (sig)--(t1)--(2.7,1);
  \draw (4,2.3)[dashed]--(4,3.6);
  \draw (r3)[solid]--(4,3.6);
  \draw (5.3,1)[dashed]--(6.6,1);
  \draw (s3)[solid]--(6.6,1);
  \draw (2.74,1)[dashed]--(1.2,1);
  \draw (1.3,1)[solid]--(t3);
\end{tikzpicture}
\end{center}
\caption{Coxeter-graph with a ringed node \cite{rp}. The graph encodes a particular figure, the ringed node distinguishes it from other figures having the same symmetry.}
\label{cdd}
\end{figure}
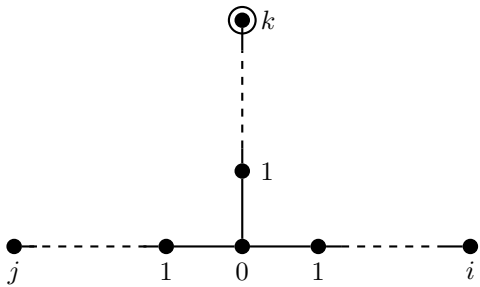

These cases establish the basic significance of the considered domain: within it, the prototypes of integrable systems in one, two and three dimensions, fit together on an equal footing.
It is therefore an example that allows to understand a natural unity between the dimensions encoded in a generalised lattice geometry.
It is complementary to the usual hierarchical view; instead of placing the lower dimensional integrable systems as reductions of those in higher dimension, they are seen as co-operating parts of a larger whole.

This extends the similar relation between quadrirational Yang-Baxter maps (KdV-type dynamics) and Painlev\'e equations \cite{ib,ay} to also include the KP case.

The canonical initial-value-problem for equation (\ref{skp}) on the introduced domain with $k=1$, formulated as a birational group, constitutes a representation of Coble's group of Cremona transformations associated with invariants of general point sets in projective space \cite{COBI,COBII}.
This group is the remarkable object, characterised geometrically, that takes a central role in the geometric framework of the elliptic Painlev\'e equation \cite{sc,10E9,kmnoy}.
The representation which emerges here, in the form of an isotropic lattice system generalising the (multidimensional) discrete KP and KdV lattices, establishes a rather direct contact between the infinite-degree-of-freedom integrable systems and the geometric framework of the Painlev\'e equations.

The structure of the paper is as follows.
In Section \ref{LG} the domain and a canonical initial-value-problem are defined in terms of the T-shaped Coxeter graph with distinguished node.
In Section \ref{BGC}, two local properties are singled-out as sufficient for the global consistency.
One is a well-known property connected with Desargues' theorem in the geometric setting for (\ref{skp}).
The other is discussed in Section \ref{MI}, it can be understood naturally by interpretation in the M\"obius group, which establishes also a broader context of Caley-Bezout-type correspondences.
The symmetry properties are the original motivation to extend previously known domains for (\ref{skp}), this is described in Section \ref{ISOTROPY} under the title {\it Isotropy}.
Identification with the discrete KP system in the $A_N$ case is made in Section \ref{KP}, and the $D_N$ case is connected with the KdV equation in Section \ref{KDV}, this involves a change of coordinates, and in the latter case also a partial integration.
In Section \ref{ECS}, connection with Coble's Cremona group is established by a change of variables.

\section{Definition of the domain}\label{LG}
The T-shaped graph in Figure \ref{cdd} was considered by Coxeter in \cite{rp} (see Chapter $\mathrm{XI}$ as well as the epilogue), it emerges from the framework allowing a unified description of the Gosset polytopes and tessellations \cite{GP}.
In that setting, the graph encodes a figure, with vertices, edges, faces, and so on.
Coxeter's representation is a standard one for this class of figures, which are also known, through abbreviation of the graph, as the $k_{ij}$ polytopes and honeycombs.

Elements of the domain for (\ref{skp}) will be defined no differently than the elements of the figure, the association between the two is as follows:
\begin{equation}
\begin{split}
{\textrm {variables}} \ &\leftrightarrow \ {\textrm {vertices}},\\
{\textrm {pair-triples}} \ &\leftrightarrow \ {\textrm {octahedral cells}}.
\end{split}
\label{fig}
\end{equation}
The pair-triple
\begin{equation}
\{\{x_1,\bar{x}_1\},
\{x_2,\bar{x}_2\},
\{x_3,\bar{x}_3\}\}, \label{loosevar}
\end{equation}
is the natural stencil for (\ref{skp}) because it shares its invariance under permutations of the variables.
\begin{definition}\label{domain}
Let integers $i,j$ be positive, and $k$ be non-negative.
Denote by $G$ the group defined by its presentation in terms of generators 
\begin{equation}
S:=\{r_0,\ldots,r_k,s_1,\ldots,s_j,t_1,\ldots,t_i\}\label{generators}
\end{equation}
and relations encoded in Figure \ref{cdd}, where generators are identified with nodes as follows:
\begin{equation}\label{dia}
\begin{array}{cccccccccccccc}
&&&&&& r_k\\
&&&&&& \vert\\
&&&&&& \vdots\\
&&&&&& \vert\\
&&&&&& r_1\\
&&&&&& \vert\\
s_j&-&\cdots&-&s_1&-&r_0&-&t_1&-&\cdots&-&t_i\\
\end{array}
\end{equation}
Thus, elements of $S$ are involutions, and any pair of them, $g,h$, commute, unless they are connected by an edge of the graph, in which case they satisfy the braid relation $ghg=hgh$.
Let $H$ denote the subgroup generated by $S\setminus\{r_k\}$, corresponding to deletion of the distinguished node in Figure \ref{cdd}.
The variables of the domain are assigned to the left cosets of $H$:
\begin{equation}
{\textrm{ variables}} \quad \leftrightarrow \quad gH : g\in G.
\end{equation}
A basic pair-triple is formed by the action of subgroup $\langle s_1,r_0,t_1 \rangle$ on the coset $X=r_0\cdots r_k H$,
\begin{equation}\label{bpt}
\{\{r_0X,r_0s_1t_1X\},\{ s_1X,t_1X\},\{ X,s_1t_1X\}\}.
\end{equation}
The remaining pair-triples are obtained from (\ref{bpt}) by the action of $G$:
\begin{equation}\label{pts}
{\textrm{pair-triples}} \quad \leftrightarrow \quad \{\{gr_0X,gr_0s_1t_1X\},\{ gs_1X,gt_1X\},\{ gX,gs_1t_1X\}\}: g\in G.
\end{equation}
\end{definition}
\begin{remark}
The domain has been specified in a combinatorial way, as a set of variables arranged into pair-triples.
Coxeter's framework allows to identify these with vertices and octahedral cells of a figure which corresponds to a polytope when the associated group is finite.
In that case, the number $N=i+j+k+1$ gives the dimension of the Euclidean space in which the polytope is found.
Even when this condition is not met, it is useful to maintain the intuitive understanding of this domain as an abstract polytope, generalising an $N$-dimensional cube or simplex, but this has the caveat of there being three natural ways in which to increment dimension, corresponding to the triple of integers $(i,j,k)$.
\end{remark}
Of fundamental significance in consideration of this domain, is the initial data set.
\begin{definition}\label{ivp}
The canonical initial data set associated with the pair-triple arrangement of Definition \ref{domain} is the subset of $i+j+k+ij+2$ variables corresponding to the cosets
\begin{equation}
\begin{split}
Y_n = r_n\cdots r_0X, & \quad n\in\{0,\ldots,k\},\\
Y_{mn} = t_m\cdots t_1 s_n\cdots s_1 X, & \quad m\in\{0,\ldots,i\},\ n\in\{0,\ldots,j\},
\end{split}\label{ivpcosets}
\end{equation}
where $X=Y_{00}=r_0\cdots r_kH$.
(To clarify the zero index case, $Y_{m0}=t_m\cdots t_1 X$ and $Y_{0n}=s_n\cdots s_1 X$.)
\end{definition}
\begin{remark}\label{sivp}
This initial data set is highly symmetric: it is an orbit of the subgroup generated by $S\setminus \{r_0\}$.
It constitutes a third relevant sub-figure extending the associations (\ref{fig}).
This is called a compound figure because the associated Coxeter-graph is not connected, it is obtained by deleting the node labelled $0$ in Figure \ref{cdd}.
\end{remark}
It is the first intention to demonstrate that Definition \ref{ivp} does indeed give a consistent initial-value-problem for equation (\ref{skp}).
To proceed, some further definitions are useful.
The analysis benefits by considering, not equation (\ref{skp}) itself, but a generic equation with the same stencil.
\begin{definition}\label{pte}
A pair-triple equation is an equation in six variables (\ref{loosevar}), which is invariant under permutations of the variables that preserve their partition into pairs indicated in (\ref{loosevar}), and which determines one variable rationally in terms of the others.
\end{definition}
Consistency of the initial-value-problem will mean the following.
\begin{definition}\label{consistent_ivp}
A pair-triple equation (Definition \ref{pte}) is called consistent when imposed on the domain in Definition \ref{domain}, if variables from the initial data set given in Definition \ref{ivp} are unconstrained, and all remaining variables are determined as composed rational functions of them.
\end{definition}
These definitions establish the setting in which local conditions for the global consistency will be considered.

\section{Birational group and consistency}\label{BGC}
In this section the initial-value-problem (Definition \ref{ivp}) for a generic pair-triple equation (Definition \ref{pte}) is formulated as a birational group.
This formulation is used to establish basic conditions for the consistency (Definition \ref{consistent_ivp}).
The analysis is formal (in essence, combinatorial), singularities of the actions are not considered.

It is natural to think of the initial data set, i.e., those variables of Definition \ref{domain} associated with cosets (\ref{ivpcosets}), as being arranged into arrays:
\begin{equation}
\left[\begin{array}{c}
y_0\\
y_1\\
\vdots\\
y_k
\end{array}\right],
\quad
\left[\begin{array}{cccc}
y_{00} & y_{10} & \cdots & y_{i0}\\
y_{01} & y_{11} & \cdots & y_{i1}\\
\vdots & \vdots & & \vdots \\
y_{0j} & y_{1j} & \cdots & y_{ij}
\end{array}\right].\label{Xdef}
\end{equation}
Here $y_n$ and $y_{mn}$ denote variables associated with cosets $Y_n$ and $Y_{mn}$, respectively, from (\ref{ivpcosets}).
\begin{remark}\label{unbounded}
The initial data array is finite whenever $i,j$ and $k$ are all finite.
The dynamical systems with an infinite number of degrees of freedom correspond to the case when at least one of these integers is considered to be infinite, and in that case the corresponding arrays will have unbounded extent.
\end{remark}
The actions are defined on these arrays.
\begin{definition}\label{actions}
Actions on variables (\ref{Xdef}) to be associated with generators, i.e., elements of $S$ (\ref{generators}), are as follows: 
\begin{equation}\label{rational_actions}
\begin{array}{rll}
r_n: & y_{n-1} \leftrightarrow y_{n}, & n\in\{1,\ldots,k\},\\
t_m: & y_{(m-1)n} \leftrightarrow y_{mn}, & m \in \{1,\ldots,i\},\  n\in\{0,\ldots,j\},\\
s_n: & y_{m(n-1)} \leftrightarrow y_{mn}, & m \in \{0,\ldots,i\},\  n\in\{1,\ldots,j\},\\
r_0: & y_0\leftrightarrow y_{00}, \ y_{mn}\rightarrow \bar{y}_{mn}, \ & m\in\{1,\ldots,i\},\  n\in\{1,\ldots,j\},
\end{array}
\end{equation}
where trivial actions are omitted, and $\bar{y}_{mn}$ is determined by a pair-triple equation (Definition \ref{pte}) imposed on the variables
\begin{equation}\label{eqpt}
\{\{y_0,\bar{y}_{mn}\},\{y_{0n},y_{m0}\},\{y_{00},y_{mn}\}\}.
\end{equation}
\end{definition}
The consistency is characterised in terms of these actions.
\begin{lemma}\label{consistent_actions}
A pair-triple equation (Definition \ref{pte}) is consistent in the sense of Definition \ref{consistent_ivp} if and only if the associated actions given in Definition \ref{actions} satisfy the relations encoded in the Coxeter-graph (\ref{dia}).
\end{lemma}
\begin{proof}
The main consideration here is the group of symmetries of the domain (Definition \ref{domain}), i.e., the permutations of the variables that also permute the pair-triples, corresponding to the left action of $G$ on the left cosets of $H$.
The proof involves two complementary constructive notions, corresponding to whether the second statement of the lemma is deduced from the first, or the converse.

To be convinced of the details, in both constructive parts, requires consideration of the set of equations
\begin{equation}\label{key}
g(w(Y))=w(gY), \quad Y\in\{Y_0,\ldots,Y_k,Y_{00},\ldots,Y_{ij}\}.
\end{equation}
Here $w(Y)$ denotes the variable associated with coset $Y$ in Definition \ref{domain}, so, in particular, $y_n=w(Y_n)$ and $y_{mn}=w(Y_{mn})$ for variables (\ref{Xdef}) associated with cosets (\ref{ivpcosets}).
What is meant by $g$ in (\ref{key}) is distinguished by context.
On the left-hand-side, $g$ denotes the action associated with elements of $S$ (\ref{generators}), that have been specified in Definition \ref{actions}. 
Whereas, on the right-hand-side, $g$ is considered as an element of the group $G$ defined by its presentation in Definition \ref{domain}.
The calculation required, is to confirm that system (\ref{key}) holds for the generators, i.e., when $g\in S$.
For this calculation, it is convenient to observe that the pair-triple (\ref{eqpt}) corresponds to the coset pair-triple obtained from the basic one (\ref{bpt}) by the left action of the particular group element $g=t_m\cdots t_1s_n\cdots s_1$.
This is the only calculation needed for the proof. 
It remains to describe the two constructions that have been mentioned.

The first construction is the induced birational representation of $G$, which relies on assuming the pair-triple equation is consistent in the sense of Definition \ref{consistent_ivp}.
A given symmetry $g\in G$ permutes the variables of the domain, in particular, sending variables of the initial data set to some similar subset of variables in the domain.
The consistency implies that the new variables, to which variables of the initial data set are mapped by $g$, can themselves be expressed uniquely in terms of the variables from the initial data set.
This determines an action on variables of the initial data set, i.e., the array (\ref{Xdef}), which is called the induced action of $g$.
The previous consideration of the system (\ref{key}), confirms that the induced actions of the generators are exactly those specified in Definition \ref{actions}.
By construction, these actions are generators for a representation of $G$, and, in particular, they satisfy the desired relations (\ref{dia}).

The second constructive notion assumes that actions in Definition \ref{actions} satisfy the desired relations (\ref{dia}).
Recall that the variables of the domain are originally defined by the left action of $G$ on the coset $H$ (Definition \ref{domain}).
Therefore, through associations (\ref{key}), the iterated actions of the generators will consistently determine all remaining variables in the domain as composed rational functions of those in (\ref{Xdef}).
And because $G$ also acts transitively on the pair-triples (Definition \ref{domain}), the consistency of the system (\ref{key}) for generators, which has been verified by calculation, means that the imposed equation (Definition \ref{pte}) holds on {\it all} pair-triples of the domain.
In other words, the pair-triple equation is consistent in the sense of Definition \ref{consistent_ivp}, and the proof is complete.
\end{proof}
A consequence of the symmetric form of the initial data set (cf. Remark \ref{sivp}) is that all but one of the actions (Definition \ref{actions}) are permutations.
This simple form of the acquired actions makes it straightforward to establish the fundamental consistency result:
\begin{thm}\label{fc}
A pair-triple equation (Definition \ref{pte}) is consistent in the sense of Definition \ref{consistent_ivp} for any positive integers, $i,j$, and non-negative integer, $k$, as a consequence of the particular cases, $(i,j,k)=(1,1,1)$, and $(i,j,k)=(2,1,0)$.
See Figure \ref{consistency_figures}.
\end{thm}
\begin{proof}
According to Lemma \ref{consistent_actions}, the global consistency is demonstrated by showing that actions (\ref{rational_actions}) satisfy the relations encoded in the Coxeter-graph (\ref{dia}).
Relations not involving generator $r_0$ are inherently satisfied.
Commutativity of the actions $r_2,\ldots,r_k$, $s_2,\ldots,s_j$, and $t_2,\ldots,t_i$, with $r_0$, rely only on the assumed pair-triple symmetry of the imposed equation.
Similarly, relation $r_0^2=\id$ is a consequence of the assumed symmetry.
It therefore remains to consider the relations $(r_0s_1)^3=(r_0t_1)^3=(r_0r_1)^3=\id$.
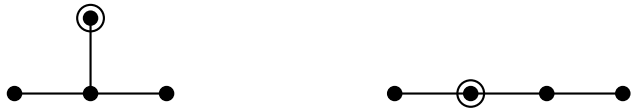
\begin{figure}
\begin{center}
\begin{tikzpicture}[thick]
\node at (1,1) [draw,circle,fill=white,minimum size=10pt,inner sep=0pt]{};
\node at (6,0) [draw,circle,fill=white,minimum size=10pt,inner sep=0pt]{};
\tikzstyle{every node}=[draw,circle,fill=black,minimum size=5pt,inner sep=0pt];
\draw (5,0) node{} -- (6,0) node{} -- (7,0) node{} -- (8,0) node{};
\draw (0,0) node{} -- (1,0) node{} -- (2,0) node{};
\draw (1,0) -- (1,1) node{};
\end{tikzpicture}
\end{center}
\caption{Coxeter-graphs which encode domains associated with the essential $(1,1,1)$ and $(2,1,0)$ self-consistency properties.}
\label{consistency_figures}
\end{figure}

Consider first the mapping $(r_0r_1)^3$.
It is clear that it inherently fixes all entries of the array on the left in (\ref{Xdef}), and so it remains to consider action on the array on the right.
Assuming all relations hold in the case $(i,j,k)=(1,1,1)$ implies that $(r_0r_1)^3$ fixes the $2 \times 2$ sub-array at the intersection of the first two rows and columns.
Because $s_2,\ldots,s_j$, and $t_2,\ldots,t_i$ inherently commute with $r_0$ and $r_1$, they also commute with $(r_0r_1)^3$. 
But $s_2,\ldots,s_j$, $t_2,\ldots,t_i$ permute the remaining entries of the array with entries on row $2$ and column $2$, and this implies that if $(r_0r_1)^3$ fixes entries of the $2\times 2$ sub-array, it must fix all remaining entries as well.
Thus $(r_0r_1)^3=\id$ in general.

Similar arguments apply to the composed mapping $(r_0s_1)^3$.
First, this composed mapping inherently fixes entries of the array on the left in (\ref{Xdef}).
Second, assuming all relations hold when $(i,j,k)=(2,1,0)$, implies that $(r_0s_1)^3$ also fixes the $3\times 2$ sub-array at the intersection of the first three rows and two columns of the array on the right in (\ref{Xdef}).
Finally, the inherent commutativity between the actions $s_3,\ldots,s_j$, $t_2,\ldots,t_i$, and $(r_0s_1)^3$, implies that all remaining rows and columns of the array must also be fixed by $(r_0s_1)^3$, and therefore $(r_0s_1)^3=\id$ in general.

Clearly if relation $(r_0s_1)^3=\id$ holds for any $(i,j,k)$, it means it also holds when the value of $i$ or $j$ is increased so that $i=j$.
In that case, an additional mapping $\omega$, which fixes the array on the left in (\ref{Xdef}), but transposes the square array on the right, can also be considered.
This corresponds to the diagram automorphism present when $i=j$, it interchanges the two undistinguished branches of the graph, $\omega s_n = t_n \omega$, $n\in\{1,\ldots,i\}$, it clearly commutes with $r_1,\ldots,r_k$, and commutativity with $r_0$ follows from the assumed symmetry of the imposed equation.
Thus, conjugation by $\omega$ shows that $(r_0t_1)^3=\id$ is a consequence of $(r_0s_1)^3=\id$, and this completes the proof.
\end{proof}
Theorem \ref{fc} singles out two local properties sufficient for the global consistency.
The pair-triple arrangement and initial-value-problem associated with each is shown in Figure \ref{local_consistency}.
\begin{figure}\label{local_consistency}
\begin{equation*}
\begin{array}{l}
\{
\{{\color{gray}{x_1}},{\color{black}{x_{023}}}\},
\{{\color{gray}{x_2}},{\color{gray}{x_{013}}}\},
\{{\color{gray}{x_3}},{\color{gray}{x_{012}}}\}
\}\\
\{
\{{\color{gray}{x_0}},{\color{black}{x_{123}}}\},
\{{\color{gray}{x_2}},{\color{gray}{x_{013}}}\},
\{{\color{gray}{x_3}},{\color{gray}{x_{012}}}\}
\}\\
\{
\{{\color{gray}{x_0}},{\color{black}{x_{123}}}\},
\{{\color{gray}{x_1}},{\color{black}{x_{023}}}\},
\{{\color{gray}{x_3}},{\color{gray}{x_{012}}}\}
\}\\
\{
\{{\color{gray}{x_0}},{\color{black}{x_{123}}}\},
\{{\color{gray}{x_1}},{\color{black}{x_{023}}}\},
\{{\color{gray}{x_2}},{\color{gray}{x_{013}}}\}
\}
\end{array}
\qquad
\begin{array}{l}
\{
\{{\color{gray}{x_{12}}},{\color{black}{x_{34}}}\},
\{{\color{gray}{x_{13}}},{\color{gray}{x_{24}}}\},
\{{\color{gray}{x_{14}}},{\color{gray}{x_{23}}}\}
\}\\
\{
\{{\color{gray}{x_{12}}},{\color{black}{x_{35}}}\},
\{{\color{gray}{x_{13}}},{\color{gray}{x_{25}}}\},
\{{\color{gray}{x_{15}}},{\color{gray}{x_{23}}}\}
\}\\
\{
\{{\color{gray}{x_{12}}},{\color{black}{x_{45}}}\},
\{{\color{gray}{x_{14}}},{\color{gray}{x_{25}}}\},
\{{\color{gray}{x_{15}}},{\color{gray}{x_{24}}}\}
\}\\
\{
\{{\color{gray}{x_{13}}},{\color{black}{x_{45}}}\},
\{{\color{gray}{x_{14}}},{\color{black}{x_{35}}}\},
\{{\color{gray}{x_{15}}},{\color{black}{x_{34}}}\}
\}\\
\{
\{{\color{gray}{x_{23}}},{\color{black}{x_{45}}}\},
\{{\color{gray}{x_{24}}},{\color{black}{x_{35}}}\},
\{{\color{gray}{x_{25}}},{\color{black}{x_{34}}}\}
\}
\end{array}
\end{equation*}
\caption{Pair-triple arrangements of Definition \ref{domain} in the cases $(i,j,k)=(1,1,1)$ and $(i,j,k)=(2,1,0)$ respectively. 
In each case, the canonical initial data (Definition \ref{ivp}) is indicated in grey.
This allows to see by inspection, that a pair-triple equation is consistent (Definition \ref{consistent_ivp}) in each case, if it is satisfied on the last two pair-triples as a consequence of being imposed on the first two pair-triples, or the first three pair-triples, respectively.
Note that the coordinatisation used generalises to the $(i,1,1)$ and $(i,j,0)$ pair-triple arrangements respectively, cf. Lemmas \ref{demicubelemma} and \ref{2030lemma}.
}
\end{figure}
In the context of the relations (\ref{fig}), these arrangements are associated with the four-dimensional cross-polytope and the rectified four-simplex, respectively.

For a given pair-triple equation, the two local consistency properties are easily investigated by calculation, but for (\ref{skp}) there is a deeper understanding available.
For the $(2,1,0)$-consistency this is a well-studied fact inherent from the geometric interpretation \cite{slg,ks1,CSC}, and is discussed in depth in \cite{DDM,ks2}.
In the following section the $(1,1,1)$-consistency will be discussed further.

\section{Connection with M\"obius involutions}\label{MI}
Unlike the $(2,1,0)$ case, the $(1,1,1)$ consistency property of (\ref{skp}) is not well known in the context of integrable systems.
An algebraic interpretation in terms of the M\"obius group offers the desirable intuition about it.

Recall that a M\"obius involution is a mapping $x\mapsto \bar{x}$ defined by a symmetric degree-$(1,1)$ polynomial, $ax\bar{x}+b(x+\bar{x})+c = 0$, where coefficients $a,b,c\in\mathbb{C}$ satisfy the condition for irreducibility $ac\neq b^2$. 
It has the feature of transposing any given point in the extended complex plane with a unique other point.
Calling such pair of points an {\it orbit} of the transformation, the following can be verified.
\begin{prop}[\cite{atk3}]\label{alg}
(i) A freely chosen disjoint pair of orbits uniquely determines a M\"obius involution.
(ii) A M\"obius involution with the three orbits (\ref{loosevar}) exists, if and only if the points in these orbits satisfy constraint (\ref{skp}).
\end{prop}
\begin{proof}
This follows by considering the linear system for the coefficients of the M\"obius involution.
It leads to a constraint on three orbits in a determinant form,
\begin{equation}\label{skp2}
\left|
\begin{array}{ccc}
x_1\bar{x}_1 & x_1+\bar{x}_1 & 1\\
x_2\bar{x}_2 & x_2+\bar{x}_2 & 1\\
x_3\bar{x}_3 & x_3+\bar{x}_3 & 1
\end{array}
\right| = 0,
\end{equation}
which is equivalent to (\ref{skp}).
Note that, in this form, the pair-triple symmetry of the equation is visible.
\end{proof}
More generally, consider the constraint on $N$ orbits of a M\"obius involution for any $N>2$.
It can be written as $N(N-1)(N-2)/6$ distinct 3-orbit constraints like (\ref{skp2}), but it is clear that all must be satisfied as a consequence of just $N-2$ of them.
Identifying $N=k+3$, this is equivalent to the consistency property corresponding to case $i=j=1$ of Definition \ref{consistent_ivp}, whose associated figure is the cross-polytope in $N$ dimensions.
In particular, the $(1,1,1)$-consistency of (\ref{skp}) is associated with the constraint on four orbits of a M\"obius involution,
\begin{equation}
\{\{x_1,\bar{x}_1\},
\{x_2,\bar{x}_2\},
\{x_3,\bar{x}_3\},
\{x_4,\bar{x}_4\}\}.
\end{equation}
\begin{remark}
The connection between equation (\ref{skp}) and the Clifford configuration \cite{ks1,ks} has led to the conformal version of the Desargues theorem in relation to the $(2,1,0)$-consistency \cite{ks2}. 
Clifford configurations are discussed in terms of M\"obius involutions in \cite{Rigby}, in that context the M\"obius involutions are called {\it half-turns of the inversive plane}, and points of an orbit have been called {\it mates}.
Although a $C4$ Clifford configuration determines four orbits of a M\"obius involution, they are not generic orbits, indeed, only three orbits are needed to determine a unique $C4$ Clifford configuration.
Therefore, despite having the same combinatorics, the $(1,1,1)$-consistency of (\ref{skp}) is not itself directly equivalent to existence of the $C4$ Clifford configuration.
It is logical that this can be resolved in a suitably generalised geometric setting.
\end{remark}
\begin{remark}
The algebraic interpretation in terms of the M\"obius involutions identifies (\ref{skp}) as the simplest case ($n=2$) in a class of systems related to the Caley-Bezout-type correspondences.
These are correspondences $x\mapsto \bar{x}$ defined by a polynomial equation $h(x,\bar{x})=0$ where
\begin{equation}\label{cbe}
h(x,\bar{x}):=\frac{r_1(x)r_2(\bar{x})-r_1(\bar{x})r_2(x)}{x-\bar{x}}.
\end{equation}
Here $r_1$ and $r_2$ are polynomials of degree $n$ with no common roots, the correspondence is $(n-1)$-valued, but its orbits inherently close, forming $n$-tuples.
The roots of $r_1$ and $r_2$ constitute two of the orbits, which may therefore be chosen freely, but there is a constraint on any set of three orbits.
The present work relates to the case $n=2$, and as such is a prequel to \cite{ib,ay} which investigated the $n=3$ family of correspondences.
That work was in fact restricted to members of the family that share a single orbit, furthermore the next case, $n=4$, is the highest value of $n$ for which the three-orbit constraint is rational, and its significance is un-known at this time.
\end{remark}

\section{Isotropy}\label{ISOTROPY}
The purpose of this section is to establish a completeness property for the introduced domain, it is a combinatorial notion based on symmetry introduced in \cite{ib}.
This property is not connected to integrability, instead it is a notion to help decide if the domain of a system has been prematurely truncated.

Recall that in general the domain of a pair-triple equation (Definition \ref{pte}) is considered as a set of variables arranged into pair-triples, and a symmetry of the domain is a permutation of the variables that also permutes the pair-triples.
In the present situation the symmetry considerations single out the more specific case that $i=j$.
\begin{prop}\label{sc}
When $i=j$, the pair-triple arrangement of Definition \ref{domain} is symmetry-complete.
That is, every symmetry of every pair-triple is the restriction of a symmetry of the whole arrangement.
\end{prop}
\begin{proof}
Consider Definition \ref{domain}.
Inherent from the way the arrangement is defined, only a single pair-triple needs to be considered, and it is natural to focus on the basic one (\ref{bpt}).
It is clear that subgroup $\langle s_1,r_0,t_1\rangle$ stabilizes this pair-triple, however, this subgroup does not give the full pair-triple (in other words, octahedral) symmetry.
This is rectified in the case $i=j$ when the arrangement has an additional symmetry corresponding to the diagram automorphism.
Specifically, the group $G$ is complemented by an additional generator, $\omega$, where relations involving $\omega$ are $\omega^2=(r_n\omega)^2=\id$, $n\in\{0,\ldots,k\}$, and $\omega t_n = s_n\omega$, $n\in\{1,\ldots,i\}$.
Both $G$ and $H$ are normalized by $\omega$, and therefore the action of conjugation by $\omega$ permutes the cosets of $H$ in $G$.
It is straightforward to verify that this action also permutes the pair-triples (\ref{pts}), and therefore gives a symmetry of the arrangement.
Furthermore,  this action of $\omega$ stabilizes the basic pair-triple (\ref{bpt}).
Unlike $\langle s_1,r_0,t_1\rangle$, the group $\langle \omega \rangle \ltimes \langle s_1,r_0,t_1\rangle$ gives the complete pair-triple symmetry.
\end{proof}
\begin{remark}
A domain which is symmetry-complete can be called {\it isotropic} in a sense appropriate for discrete systems.
It is the condition that, viewed from an elementary cell, the remainder of the domain looks the same in all directions.
\end{remark}
\section{The $A_N$ KP lattice}\label{KP}
In this section the $A_N$ KP lattice is discussed.
Its theory has been developed as the multidimensional extension of the face-centred-cubic lattice in \cite{DAWG,ABSo,ks2}.
Connection with instance $k=0$ of Definition \ref{domain} is established simply by changing to standard coordinates, the canonical initial-value-problem (Definition \ref{ivp}) is also re-expressed in these coordinates.
In the second part of this section, it is established that the $A_N$ lattice can be be viewed as a symmetry-complete extension of the `hypercube' domain for (\ref{skp}).
\subsection{Coordinatisation}
In the case $k=0$, the figure associated with the Coxeter graph of Figure \ref{cdd} is the $j$-rectified $N$-simplex, where $N=i+j+1$.
A natural coordinatisation is available by associating variables with the subsets of size $j+1$ in a set of size $N+1$.
\begin{definition}\label{2030}
Let $i,j$ be positive integers and $N=i+j+1$.
The pair-triple arrangement of the $j$-rectified $N$-simplex is a set of $(N+1)!/[(i+1)!(j+1)!]$ variables
\begin{equation}
x_J, \quad J\subset \{1,\ldots,N+1\},\ |J|=j+1, \label{vars}
\end{equation}
arranged into $(N+1)!/[(i-1)!(j-1)!4!]$ pair-triples,
\begin{equation}\label{pt}
\{\{x_{K\cup \{a,b\}},x_{K\cup\{c,d\}}\},
\{x_{K\cup \{a,c\}},x_{K\cup\{b,d\}}\},
\{x_{K\cup \{a,d\}},x_{K\cup\{b,c\}}\}\},
\end{equation}
where $|K|=j-1$ and $a,b,c,d\in\{1,\ldots,N+1\}\setminus K$ with $|\{a,b,c,d\}|=4$.
\end{definition}
The symmetry group of this arrangement acts naturally by free permutation of indices $1,\ldots,N+1$, which is the point of contact with the coordinatisation given in Definition \ref{domain}.
\begin{lemma}\label{2030lemma}
The pair-triple arrangement in case $k=0$ of Definition \ref{domain}, coincides with the $j$-rectified $N$-simplex arrangement (Definition \ref{2030}).
\end{lemma}
\begin{proof}
Consider case $k=0$ of Definition \ref{domain}.
A consistent and faithful representation of $G$ as free permutation of some set of indices $1,\ldots,N+1$ can be introduced by associating actions with the generators as follows:
\begin{equation}\label{intactions}
\begin{split}
s_n & : j+1-n\leftrightarrow j+2-n, \quad n\in\{1,\ldots,j\},\\
r_0 & : j+1\leftrightarrow j+2,\\
t_n & : j+1+n\leftrightarrow j+2+n, \quad n\in\{1,\ldots,i\}.
\end{split}
\end{equation}
In this representation, $H$ is characterised as the subgroup which stabilises the index-set $\{1,\ldots,j+1\}$.
More generally, the image of this subset by the action of a left coset of $H$ is single-valued, so there is a bijection
\begin{equation}\label{otherpart}
gH \ \leftrightarrow \ \{g(1),\ldots,g(j+1)\}, \quad g\in G.
\end{equation}
This gives the coordinatisation of Definition \ref{2030}.
\end{proof}

In the coordinates of the $j$-rectified $N$-simplex, the initial data (\ref{Xdef}), associated with cosets (\ref{ivpcosets}), corresponds to the following subset of variables (\ref{vars}):
\begin{equation}
x_{\{1,\ldots,j+1\}}=y_0, \quad 
x_{(\{1,\ldots,j+1\}\setminus\{j+1-n\})\cup\{j+2+m\}} = y_{mn},
\end{equation}
where $m\in\{0,\ldots,i\}$ and $n\in\{0,\ldots,j\}$.
\begin{remark}\label{Schief}
It is clear, from the proof of Theorem \ref{fc}, that the $(1,1,1)$-consistency of (\ref{skp}) plays no role for consistency on the $A_N$ sub-lattice.
Sufficiency of the $(2,1,0)$-consistency in this setting was established originally in \cite{ks2}.
It contrasts with the considerations in \cite{ABSo} where consistency on a larger stencil was required, but this is explained by a generalised setting involving no symmetry assumptions of the imposed equations.
Also it can be compared with the approach in \cite{DDM}, where the consistency is presented as being inherent from the axioms of projective geometry, so the properties of equations are seen as inherited.
In this paper, the question of consistency has been approached by converting it to the problem of verifying the defining relations of a birational group.
\end{remark}

\begin{remark}\label{FCClattice}
The Face-centred-cubic (FCC) lattice, on which equation (\ref{skp}) defines a discrete version of the three-dimensional KP dynamics, is embedded naturally in the limiting case of the $j$-rectified $N$-simplex domain as $j$ and $N$ tend to infinity. 
This limiting case can be called, after its automorphism group, the $A_{\infty}$ domain.
To illustrate the embedding of the FCC sub-lattice into the $A_{\infty}$ domain, a combinatorial description of it follows.

Consider Definition \ref{2030}, and, in the first instance, suppose that $i$ and $j$ are finite.
Introduce an arbitrary partition of the index-set, $\{1,\ldots,N\}=A_1\cup A_2 \cup A_3 \cup A_4$, and endow each part, $A_\mu$, $\mu\in\{1,2,3,4\}$, with an ordering.
The ordering is chosen freely for each $A_\mu$, but is henceforth fixed.
The basic restriction is to variables $x_J$ (\ref{vars}) whose index-set $J$ is also in the form of a partition $J=J_1\cup J_2 \cup J_3 \cup J_4$, where $J_\mu\subseteq A_\mu$, and, if not empty, then $J_\mu$ contains the first element of $A_\mu$, and is itself sequential.
Restricted in this way, the subset $J$ is completely determined by the four integers $a=|J_1|$, $b=|J_2|$, $c=|J_3|$ and $d=|J_4|$, which, due to the constraint $|J|=j+1$, are themselves subject to the condition $a+b+c+d=j+1$.
Replacing index-set $J$ with the four integers that determine it, the restricted subset of variables are specified as
\begin{equation}\label{fccvars}
x_{a,b,c,d}, \quad a+b+c+d=j+1.
\end{equation}
The pair-triples from (\ref{pt}) that contain wholly variables from the restricted set (\ref{fccvars}), are straightforwardly found to be
\begin{multline} \label{fccpt}
\{\{x_{a+1,b+1,c,d},x_{a,b,c+1,d+1}\},\\
\{x_{a+1,b,c+1,d},x_{a,b+1,c,d+1}\},
\{x_{a+1,b,c,d+1},x_{a,b+1,c+1,d}\}\},
\end{multline}
where $a+b+c+d=j-1$.
In (\ref{fccvars}) the integers $a,b,c$ and $d$ are non-negative, and each will not exceed the size of $A_1$, $A_2$, $A_3$ and $A_4$ respectively.
Similarly in (\ref{fccpt}), but the upper limit is set, respectively, by $|A_\mu|-1$.

Variables (\ref{fccvars}) arranged into the pair-triples (\ref{fccpt}), define, combinatorially, vertices and octahedral cells of a connected finite part of a FCC lattice.
The extent of the embedded part can always be increased by increasing the value of integers $i$ and $j$: the infinite FCC lattice is therefore embedded in the limiting case that $i$ and $j$ are incremented indefinitely.
From the point of view of this embedding, the consequence of the $(2,1,0)$-consistency (cf. Remark \ref{Schief}) is that the generic solution of equation (\ref{skp}) on the FCC lattice can always be extended to a solution on the $A_\infty$ lattice, which is known as the multidimensional extension.
\end{remark}

\subsection{Restriction to the hypercube}\label{HYPERCUBE}
An important domain for (\ref{skp}) is the hypercube of arbitrary dimension, it is natural in the context of nonlinear superposition \cite{NCWQ,DN}, where commuting B\"acklund transformations are associated with lattice edges.
\begin{definition}\label{Ncube}
For a fixed integer $M>2$, the pair-triple arrangement of the $M$-cube is a set of $2^M-2$ variables
\begin{equation}
x_{\bar{J}}, \quad \{\}\subset \bar{J} \subset \{1,\ldots,M\},\label{cubevars}
\end{equation}
arranged into $2^{M-3}M(M-1)(M-2)/6$ pair-triples
\begin{equation}\label{cubept}
\{\{x_{\bar{K}\cup\{n\}}, x_{\bar{K}\cup \{m,l\}}\},
\{x_{\bar{K}\cup\{m\}}, x_{\bar{K}\cup \{l,n\}}\},
\{x_{\bar{K}\cup\{l\}}, x_{\bar{K}\cup \{n,m\}}\}\},
\end{equation}
where $\bar{K}\subset\{1,\ldots,M\}$ with $|\bar{K}|<M-2$ and $n,m,l\in\{1,\ldots,M\}\setminus \bar{K}$ such that $|\{n,m,l\}|=3$.
\end{definition}
Each multi-index $\bar{J}$ is associated with a vertex of the hypercube, and each pair-triple is associated with a cube.
More precisely, each pair-triple is associated with a cube that has a distinguished diagonal. 
Such restriction of the cube does not have the full pair-triple symmetry, and it is therefore easily seen that the arrangement of Definition \ref{Ncube} is not symmetry-complete (recall Section \ref{ISOTROPY}).
The purpose of this sub-section is to establish that the $A_N$ lattice is a symmetry-complete extension of it:
\begin{prop}\label{Ncubereduced}
The pair-triple arrangement of the $j$-rectified $N$-simplex (Definition \ref{2030}) in case $i=j$, is a symmetry-complete extension of the arrangement of the $M$-cube (Definition \ref{Ncube}), where $i=j=M-2$ and $N=2M-3$.
\end{prop}
\begin{proof}
Due to Proposition \ref{sc}, it is sufficient to show that the case $i=j$ of the $j$-rectified $N$-simplex arrangement (Definition \ref{2030}) contains a subset of variables and pair-triples which form an $M$-cube arrangement (Definition \ref{Ncube}) with $M=j+2$.
Let $i=j=M-2$ and introduce a set $A$:
\begin{equation}
A=\{M+1,\ldots,2M-2\}.
\end{equation}
In case $i=j=M-2$, the variables (\ref{vars}) of arrangement in Definition \ref{2030} are associated with size-$(M-1)$ subsets of $\{1,\ldots,2M-2\}$, the key restriction is to subsets whose intersection with $A$ is sequential:
\begin{equation}\label{varsrestrict}
J\cap A = \{M+1,\ldots,M+|J\cap A|\}.
\end{equation}
It is clear that each such $J$ is uniquely determined by the part not contained in $A$, $J\setminus A$.
Identification $\bar{J}=J\setminus A$ defines a correspondence between variables (\ref{vars}) restricted by (\ref{varsrestrict}), and the variables (\ref{cubevars}).
The corresponding restriction on the set of pair-triples (\ref{pt}) is that participating subsets $K$ and $\{a,b,c,d\}$ satisfy the additional constraints
\begin{equation}
K\cap A = \{M+1,\ldots,M+|K\cap A|\},\label{c1}
\end{equation}
and
\begin{equation}
(K\cup \{a,b,c,d\})\cap A = \{M+1,\ldots,M+|K\cap A|+1\}.\label{c2}
\end{equation}
Notice that this implies exactly one element of $\{a,b,c,d\}$ is contained in $A$.
It can then be verified, identifying $\bar{K}=K\setminus A$ and $\{n,m,l\}=\{a,b,c,d\}\setminus A$, that there is a one-to-one correspondence between pair-triples (\ref{pt}) restricted by (\ref{c1}), (\ref{c2}), and the pair-triples (\ref{cubept}).
\end{proof}

\section{The $D_N$ KdV lattice}\label{KDV}
This section considers a $D_N$ sub-lattice, corresponding to case $j=k=1$ of Definition \ref{domain}.
The purpose is to establish that equation (\ref{skp}) imposed on this domain is equivalent to the discrete KdV equation \cite{WE,hirota-0} in its well-known Schwarzian guise \cite{NC,Nij2,James2} and multidimensional setting \cite{NW,BS1}.

The associated figure is the demihypercube in dimension $N=i+3$.
The transition to natural demihypercube coordinates can be made as follows.
\begin{lemma}\label{demicubelemma}
The domain of Definition \ref{domain} in the case $k=j=1$ can be written alternatively as a set of variables 
\begin{equation}
u_I, \quad I\subseteq\{0,\ldots,i+2\},\ |I| {\textrm{ odd}},
\end{equation}
arranged into pair-triples
\begin{equation}\label{demicubetp}
\{\{u_{I\oplus\{a\}},u_{I\oplus\{b,c,d\}}\},
\{u_{I\oplus\{b\}},u_{I\oplus\{c,d,a\}}\},
\{u_{I\oplus\{c\}},u_{I\oplus\{d,a,b\}}\}\},
\end{equation}
where $I\subseteq\{0,\ldots,i+2\}$, $|I|$ even, $a,b,c,d\in\{0,\ldots,i+2\}$, $|\{a,b,c,d\}|=4$, and $\oplus$ denotes symmetric difference, $I\oplus J=(I\cup J) \setminus (I \cap J)$.
\end{lemma}
\begin{proof}
Let $k=j=1$ and consider the group $G$ and subgroup $H$ of Definition \ref{domain}.
Define group elements $\sigma_1,\ldots,\sigma_{i+2}$ as
\begin{equation}
\begin{split}
\sigma_1 & := s_1r_1,\\
\sigma_2 & := (s_1r_1)^{r_0},\\
\sigma_{n+2} & := (s_1r_1)^{t_n\cdots t_1 r_0}, \qquad n\in\{1,\ldots,i\}.
\end{split}
\end{equation}
Then $\sigma_n\sigma_m=\sigma_m\sigma_n$ for all $n,m\in\{1,\ldots,i+2\}$, and the left action of generators (\ref{generators}) on the cosets
\begin{equation}\label{demicosets}
\sigma_IH, \quad \sigma_I:=\prod_{n\in I\setminus\{0\}} \sigma_n, \quad I\subseteq\{0,1,\ldots,i+2\},\ |I|{\textrm{ odd}},
\end{equation}
reduce to actions on $I$ as follows:
\begin{equation}\label{dca}
\begin{array}{rl}
r_1: & I\rightarrow I_{0\leftrightarrow 1},\\
r_0: & I\rightarrow I_{1\leftrightarrow 2},\\
t_n: & I\rightarrow I_{n+1\leftrightarrow n+2}, \quad n\in\{1,\ldots,i\},\\
s_1r_1: & I\rightarrow I\oplus \{0,1\}.\\
\end{array}
\end{equation}
This can be verified directly using the defining group relations.
It shows that the left cosets of $H$ in $G$ are exhausted by (\ref{demicosets}), allowing to replace the coset representation with the representation by index-sets $I\subseteq\{0,\ldots,i+2\}$, $|I|$ odd.
Note that in (\ref{dca}) the action of the composition $s_1r_1=\sigma_1$ is given instead of $s_1$ directly, $\sigma_n$ corresponds to the demihypercube reflection $I\mapsto I\oplus\{0,n\}$, and index-sets $I$ correspond to demihypercube vertices.

Introducing notation $u_I$ for variable associated with coset $\sigma_IH$ (defined in (\ref{demicosets})), the basic pair-triple (\ref{bpt}) corresponds to
\begin{equation}\label{dbpt}
\{\{u_{\{1\}},u_{\{0,2,3\}}\},
\{u_{\{0,1,2\}},u_{\{3\}}\},
\{u_{\{2\}},u_{\{0,1,3\}}\}\}.
\end{equation}
The pair-triple arrangement of Definition \ref{domain} in the case $k=j=1$ is therefore written alternatively as the set of pair-triples generated from the basic one (\ref{dbpt}) by the actions (\ref{dca}), which is exactly (\ref{demicubetp}).
\end{proof}
The set of variables in the canonical initial data set (\ref{Xdef}) (associated with cosets (\ref{ivpcosets})), expressed using the adopted notation $u_I$ for variable associated with coset $\sigma_IH$ (\ref{demicosets}), correspond as follows:
\begin{eqnarray}
y_1,y_0,y_{00},y_{10},\ldots,y_{i0} &=& u_{\{0\}},u_{\{1\}},u_{\{2\}},u_{\{3\}},\ldots,u_{\{i+2\}},\label{va1}\\
y_{01},y_{11},\ldots,y_{i1} &=& u_{\{0,1,2\}},u_{\{0,1,3\}},\ldots,u_{\{0,1,i+2\}}.\label{va2}
\end{eqnarray}

Equation (\ref{skp}) imposed on the domain of Lemma \ref{demicubelemma} is a form of the discrete Schwarzian KdV system which is essentially known going back to \cite{NQC,NCWQ}.
It can be viewed as a reduction from the hypercube domain (Section \ref{HYPERCUBE}), either as a periodic 2-cycle \cite{ks1} or, which is equivalent, by imposing additional symmetry \cite{HKM,NS}.
The more usual form arises from a partial integration which introduces the lattice parameters.
This is performed here by exploiting the associated birational group, in terms of which the following is easily verified by calculation:
\begin{lemma}
Let $k=j=1$ and consider the actions on variables (\ref{Xdef}) given in Definition \ref{actions} for underlying pair-triple equation (\ref{skp}).
Through associations
\begin{equation}\label{alphadef}
\frac{\alpha_{n+2}-\alpha_1}{\alpha_0-\alpha_1} = \frac{(y_{n1}-y_1)(y_{n0}-y_0)}{(y_{n1}-y_{n0})(y_1-y_0)}, \quad n\in\{0,\ldots,i\},
\end{equation}
they are consistent with actions on $\alpha_0,\ldots,\alpha_{i+2}$ as follows:
\begin{equation}\label{alphaactions}
\begin{array}{rll}
r_0: & \alpha_1\leftrightarrow \alpha_2,\\
r_1: & \alpha_0\leftrightarrow \alpha_1,\\
s_1: & \alpha_0\leftrightarrow \alpha_1,\\
t_n: & \alpha_{n+1}\leftrightarrow \alpha_{n+2}, & n\in\{1,\ldots,i\},
\end{array}
\end{equation}
where the trivial actions are omitted.
\end{lemma}
Equations (\ref{alphadef}) allow the subset of variables in the initial data set (\ref{va2}), to be replaced by {\it lattice parameters}, $\alpha_0,\ldots,\alpha_{i+2}$, on which the group acts not rationally, but by pure permutation.
The resulting partially-integrated system is generated by the group permutation action (Definition \ref{domain} and (\ref{alphaactions})) from the basic set of equations (\ref{alphadef}). 
Written in the demihypercube variables (Lemma \ref{demicubelemma}), this is recognisable as the lattice Schwarzian KdV system:
\begin{prop}\label{SKDVdemicube}
Equation (\ref{skp}) imposed on the domain of Definition \ref{domain} in the case $k=j=1$, is equivalent to the system
\begin{equation}\label{skdv2}
\frac{\alpha_n-\alpha_l}{\alpha_m-\alpha_l}=\frac{(u_{I\oplus \{n,m\}}-u_{I\oplus\{m,l\}})(u_{I\oplus \{n,l\}}-u_{I})}{(u_{I\oplus \{n,m\}}-u_{I\oplus\{n,l\}})(u_{I\oplus \{m,l\}}-u_{I})},
\end{equation}
where $u_I$ denotes variable associated with coset $\sigma_IH$ (\ref{demicosets}),
\begin{equation}\label{demicube}
I\subseteq\{0,\ldots,i+2\}, \ |I|{\textrm{ odd}}, \quad n,m,l\in\{0,\ldots,i+2\}, \ |\{n,m,l\}|=3,
\end{equation}
and $\alpha_0,\ldots,\alpha_{i+2}$ are free parameters.
\end{prop}
\begin{remark}\label{dc}
The system (\ref{skdv2}), (\ref{demicube}) is in coordinates of the demihypercube in $N=i+3$ dimensions.
It is better-known in coordinates of the hypercube in one lower dimension: the hypercube coordinates are recovered by deleting all occurrences of the index $0$ and setting $\alpha_0=0$, which can be done without losing generality (it lifts the restriction that $|I|$ is odd).
However, this leaves two kinds of equations.
If $0\in\{n,m,l\}$, then the equation is associated with a quad face of the hypercube, whereas if $0\not\in\{n,m,l\}$, then it is associated with a tetrahedron formed by four vertices of an embedded cube.
The equations on tetrahedra have played a significant role in the theory of this and similar quad-equations \cite{ABS,ABS2,Boll}.
In the demihypercube coordinates, both kinds of equations are accounted for together.
\end{remark}

\begin{remark}\label{N2restriction}
The transition from the canonical initial-value-problem on the hypercube domain for quad-equations, to a general planar quad-graph domain embedded within it, has been explained in \cite{AdVeQ}.
A simple example of this, illustrating the two-dimensional discrete dynamics contained in the system of Proposition \ref{SKDVdemicube}, is a description of the regular $\mathbb{N}^2$ sub-domain, which proceeds as follows.

Consider Proposition \ref{SKDVdemicube} in the case that $i=\infty$. The index subsets $I\subseteq\{1,\ldots,i+2\}$ should, however, be finite, so that meaning can be given to the condition that $|I|$ is odd.
Now introduce a partition of the index-set $\{0,\ldots,i+2\}=A_1\cup A_2 \cup \{0\}$ in which $A_1$ and $A_2$ are each infinite in extent.
Furthermore, define an arbitrary ordering of the two main parts, $A_1$ and $A_2$, so that each can be considered as a sequence with first element.
A subset of the variables $u_I$, $I\subset \{0,\ldots,i+2\}$, $|I|$ odd, are defined by the constraint that $I$ should be in the form of a partition $I=I_1\cup I_2\cup B$, where $B\in\{\{\},\{0\}\}$ is chosen so that $|I|$ is odd, whilst, for $\mu\in\{1,2\}$, $I_\mu\subset A_\mu$ is either empty, or contains the first element of $A_\mu$, and is itself sequential.
Index-sets $I$ restricted in this way are determined completely by the pair of integers $n=|I_1|$ and $m=|I_2|$, which allows to introduce the notation
\begin{equation}\label{N2vars}
u_{n,m}, \quad n,m\in\mathbb{N},
\end{equation}
for the restricted set of variables.
For convenience, set the parameter $\alpha_0=0$ without loss of generality. 
Relabel the parameters $\alpha_1,\ldots,\alpha_{i+2}$, replacing them by $p_1,p_2,\ldots$ associated sequentially with the subset of indices contained in $A_1$, and $q_1,q_2,\ldots$ associated sequentially with the remaining indices that are contained in $A_2$.
With this relabelling, the subset of equations (\ref{skdv2}) which involve only the restricted set of variables (\ref{N2vars}), take the form
\begin{equation}
\frac{p_n}{q_m} = \frac{(u_{n+1,m+1}-u_{n,m+1})(u_{n+1,m}-u_{n,m})}{(u_{n+1,m+1}-u_{n+1,m})(u_{n,m+1}-u_{n,m})},
\end{equation}
where $n,m\in\mathbb{N}$.
The two-dimensional discrete dynamics defined by this sub-system is not additionally constrained by embedding it in the infinite-dimensional demihypercube system.
This is the well-known multidimensional consistency, which encodes the integrability, and most directly the B\"acklund transformations of the system.
\end{remark}

\section{The Cremona system}\label{ECS}
Case $k=1$ of Definition \ref{actions} gives generators for a representation of the group of Coble \cite{COBI,COBII}.
This group is a basic element in the geometric framework of the Painlev\'e equations \cite{sc,10E9,WE10}, and is a canvas for several extensions \cite{8PC,kmnoy,tt}.

The connection is made by a change of variables.
\begin{prop}
Let $k=1$.
Through associations
\begin{equation}\label{uvar}
u_{mn} = \frac{(y_{mn}-y_1)(y_{00}-y_1)(y_{m0}-y_0)(y_{0n}-y_0)}{(y_{mn}-y_0)(y_{00}-y_0)(y_{m0}-y_1)(y_{0n}-y_1)},
\end{equation}
$m\in\{1,\ldots,i\}$, $n\in\{1,\ldots,j\}$, the actions of Definition \ref{actions} obtained from pair-triple equation (\ref{skp}), induce the actions
\begin{equation}\label{painleve-actions}
\begin{array}{rll}
r_1:&u_{mn} \rightarrow 1/u_{mn}, & m\in\{1,\ldots,i\}, \ n\in\{1,\ldots,j\}, \\
r_0:&u_{mn} \rightarrow 1-u_{mn}, & m\in\{1,\ldots,i\}, \ n\in\{1,\ldots,j\}, \\
t_1:&u_{1n} \rightarrow 1/u_{1n}, u_{mn} \rightarrow u_{mn}/u_{1n}, & m\in\{2,\ldots,i\}, \ n\in\{1,\ldots,j\},\\
t_m:&u_{(m-1)n} \leftrightarrow u_{mn}, & m\in\{2,\ldots,i\}, \ n\in\{1,\ldots,j\},\\
s_1:&u_{m1} \rightarrow 1/u_{m1}, u_{mn} \rightarrow u_{mn}/u_{m1}, & m\in\{1,\ldots,i\}, \ n\in\{2,\ldots,j\},\\
s_n:&u_{m(n-1)} \leftrightarrow u_{mn}, & m\in\{1,\ldots,i\}, \ n\in\{2,\ldots,j\},
\end{array}
\end{equation}
on variables $u_{11},\ldots,u_{ij}$.
\end{prop}
Actions (\ref{painleve-actions}) coincide with those tabulated in \cite{kmnoy}, the original form used by Coble (see \cite{COBI} \S 7 (64) and \cite{COBII} \S 4 (15)) are normalised slightly differently.
In \cite{kmnoy} an elliptic-function substitution associates linear actions compatible with (\ref{painleve-actions}). 
The substitution does not lose generality when $(i,j,k)=(5,2,1)$, and the equivariant extension to $(i,j,k)=(6,2,1)$ is what introduces dependent variables of the Painlev\'e equation.
\begin{remark}
The multi-ratio form (i.e., the actions on $y_0$,$y_1$,$y_{00}$,$\ldots$,$y_{ij}$) can be compared with the tau-function formulation of (\ref{painleve-actions}) that has been given in \cite{kmnoy}, and the two kinds of variables should be naturally related without recourse to the intermediary variables $u_{11},\ldots,u_{ij}$ that arise from a comparatively arbitrary normalisation.
One feature of the multi-ratio form is the inherent extension to cases other than $k=1$, and although this exceeds the group considered originally by Coble, birational groups associated with a general T-shaped Coxeter graph are known in the algebraic-geometry context \cite{Looj,sm}, which provide candidates for attributing the desirable geometric meaning to the variables $y_0,\ldots,y_k,y_{00},\ldots,y_{ij}$.
The theory of the multi-ratio representation is not developed yet to the point of answering the questions about geometric interpretation and tau-function formulation.
The obligation in developing this theory, is to ensure that the corresponding concepts known in the KP and KdV settings are respected when restriction is made to the $A_N$ and $D_N$ sub-domains respectively.
However, it holds interesting possibilities, for instance a non-commutative generalisation of (\ref{skp}) that has the same stencil, has been proposed in \cite{CSC}.
\end{remark}

\vskip5mm
\section*{Acknowledgement}
This work is supported in part by the Australian Research Council, Discovery Grant DP 110104151.
I would like to express my gratitude to M. Noumi for some encouragement and to Y. Yamada for some instruction on birational groups related to the Painlev\'e equations.
\bibliographystyle{unsrt}
\bibliography{references}

\begin{thebibliography}{10}

\bibitem{DN}
Y.~Dorfman and F.~W. Nijhoff.
\newblock On a $(2+1)$-dimensional version of the {K}richever-{N}ovikov
  equation.
\newblock {\em Phys. Lett. A}, 157:107--112, 1991.

\bibitem{slg}
W.~K. Schief.
\newblock Lattice geometry of the discrete {D}arboux, {K}{P}, {B}{K}{P}, and
  {C}{K}{P} equations. {M}enelaus' and {C}arnot's theorem.
\newblock {\em J. Nonlinear Math. Phys}, 10:194--208, 2003.

\bibitem{HirotaKP}
R.~Hirota.
\newblock Discrete analogue of a generalized {T}oda equation.
\newblock {\em J. Phys. Soc. Jpn.}, 50:3785–3791, 1981.

\bibitem{NCWQ}
F.~W. Nijhoff, H.~W. Capel, G.~L. Wiersma, and G.~R.~W. Quispel.
\newblock B\"acklund transformations and three-dimensional lattice equations.
\newblock {\em Phys. Lett. A}, 105(6):267--272, 1984.

\bibitem{DDM}
A.~Doliwa.
\newblock {D}esargues maps and the {H}irota-{M}iwa equation.
\newblock {\em Proc. R. Soc. A}, 466:1177--1200, 2010.

\bibitem{DAWG}
A.~Doliwa.
\newblock The affine {W}eyl group symmetry of the {D}esargues maps and of the
  non-commutative {H}irota-{M}iwa system.
\newblock {\em Phys. Lett. A}, 375:1219--1224, 2011.

\bibitem{ABSo}
V.~Adler, A.~Bobenko, and Y.~Suris.
\newblock Classification of integrable discrete equations of octahedron type.
\newblock {\em International Mathematics Research Notices}, 2012, No.
  8:1822–1889, 2012.
\newblock doi: 10.1093/imrn/rnr083.

\bibitem{ks2}
A.~D. King and W.~K. Schief.
\newblock Clifford lattices and a conformal generalization of {D}esargues'
  theorem.
\newblock {\em Journal of Geometry and Physics}, 62(5):1088--1096, 2012.

\bibitem{rp}
H.~S.~M. Coxeter.
\newblock {\em {R}egular {P}olytopes}.
\newblock Dover books on mathematics. Dover, 1963.

\bibitem{NC}
F.~W. Nijhoff and H.~W. Capel.
\newblock The discrete {K}orteweg-de {V}ries equation.
\newblock {\em Acta Appl. Math.}, 39(1-3):133--158, 1995.

\bibitem{NW}
F.~W. Nijhoff and A.~J. Walker.
\newblock The discrete and continuous {P}ainlev{\'e} {VI} hierarchy and the
  {G}arnier systems.
\newblock {\em Glasgow Math. J.}, 43A:109--123, 2001.

\bibitem{BS1}
A.~I. Bobenko and Yu.~B. Suris.
\newblock Integrable systems on quad-graphs.
\newblock {\em Intl. Math. Res. Notices}, 11:573--611, 2002.

\bibitem{sc}
H.~Sakai.
\newblock Rational surfaces associated with affine root systems and geometry of
  the painlev{\'e} equations.
\newblock {\em Communications in Mathematical Physics}, 220(1):165--229, 2001.

\bibitem{org}
Y.~Ohta, A.~Ramani, and B.~Grammaticos.
\newblock An affine {W}eyl group approach to the eight-parameter discrete
  {P}ainlev\'e equation.
\newblock {\em J. Phys. A: Math. Gen.}, 34:10523--10532, 2001.

\bibitem{ib}
J.~Atkinson.
\newblock Idempotent biquadratics, {Y}ang-{B}axter maps and birational
  representations of {C}oxeter groups.
\newblock {\em arxiv:1301.4613 [nlin.SI]}, 2013.

\bibitem{ay}
J.~Atkinson and Y.~Yamada.
\newblock {Q}uadrirational {Y}ang-{B}axter maps and the $\tilde{E}_8$
  {P}ainlev\'e lattice.
\newblock {\em arXiv:1405.2745 [nlin.SI]}, 2014.

\bibitem{COBI}
A.~B. Coble.
\newblock Points sets and allied {C}remona groups (part {I}).
\newblock {\em Transactions of the {A}merican {M}athematical {S}ociety},
  16(2):155--198, 1915.

\bibitem{COBII}
A.~B. Coble.
\newblock Point sets and allied {C}remona groups (part {I}{I}).
\newblock {\em Transactions of the {A}merican {M}athematical {S}ociety},
  17(3):345--385, 1916.

\bibitem{10E9}
K.~Kajiwara, T.~Masuda, M.~Noumi, Y.~Ohta, and Y.~Yamada.
\newblock $_{10}{E}_{9}$ solution to the elliptic {P}ainlev\'e equation.
\newblock {\em J. Phys. A: Math. Gen.}, 36:L263--L272, 2003.

\bibitem{kmnoy}
K.~Kajiwara, T.~Masuda, M.~Noumi, Y.~Ohta, and Y.~Yamada.
\newblock Point configurations, {C}remona transformations and the elliptic
  difference {P}ainlev\'e equation.
\newblock {\em S{\'e}minaires et Congr{\'e}s}, 14:169--198, 2006.

\bibitem{GP}
T.~Gosset.
\newblock On the regular and semi-regular figures in space of $n$ dimensions.
\newblock {\em Messenger of Mathematics}, XXIX, 1900.

\bibitem{ks1}
B.~G. Konopelchenko and W.~K. Schief.
\newblock {M}enelaus' theorem, {C}lifford configurations and inversive geometry
  of the {S}chwarzian {K}{P} hierarchy.
\newblock {\em Journal of Physics A: Mathematical and General}, 35(29):6125,
  2002.

\bibitem{CSC}
W.~K. Schief and B.~G. Konopelchenko.
\newblock A novel generalization of {C}lifford's classical point–circle
  configuration. {G}eometric interpretation of the quaternionic discrete
  {S}chwarzian {K}adomtsev-{P}etviashvili equation.
\newblock {\em Proc. R. Soc. A}, 465:1291--1308, 2009.

\bibitem{atk3}
J.~Atkinson.
\newblock Integrable lattice equations: {C}onnection to the {M}\"obius group,
  {B}\"acklund transformations and solutions.
\newblock {\em {P}h{D} thesis, The University of Leeds}, 2008.

\bibitem{ks}
A.~D. King and W.~K. Schief.
\newblock Tetrahedra, octahedra and cubo-octahedra: integrable geometry of
  multi-ratios.
\newblock {\em J. Phys. A: Math. Gen.}, 36:785–802, 2003.

\bibitem{Rigby}
J.~F. Rigby.
\newblock Half-turns and {C}lifford configurations in the inversive plane.
\newblock {\em J. London Math. Soc.}, 15(2):521--533, 1977.

\bibitem{WE}
H.~D. Wahlquist and F.~B. Estabrook.
\newblock {B}{\"a}cklund transformation for solutions of the {K}orteweg-de
  {V}ries equation.
\newblock {\em Phys. Rev. Lett.}, 31:1386--1390, 1973.

\bibitem{hirota-0}
R.~Hirota.
\newblock Nonlinear partial difference equation {I}: {A} difference analogue of
  the {K}d{V} equation.
\newblock {\em J. Phys. Soc. Jp.}, 43:1429--1433, 1977.

\bibitem{Nij2}
F.~W. Nijhoff.
\newblock On some ``{S}chwarzian equations'' and their discrete analogues.
\newblock In A.~S. Fokas and I.~M. Gel'fand, editors, {\em Algebraic Aspects of
  Integrable Systems: {I}n memory of {I}rene {D}orfman}, pages 237--260.
  Birkh\"auser Verlag, 1996.

\bibitem{James2}
J.~Atkinson and N.~Joshi.
\newblock The {S}chwarzian variable associated with discrete {K}d{V}-type
  equations.
\newblock {\em Nonlinearity}, 25(6):1851--1866, 2010.

\bibitem{NQC}
F.~W. Nijhoff, G.~R.~W. Quispel, and H.~W. Capel.
\newblock Direct linearization of nonlinear difference difference equations.
\newblock {\em Phys. Lett. A}, 97:125--128, 1983.

\bibitem{HKM}
M.~Hay, K.~Kajiwara, and T.~Masuda.
\newblock {B}ilinearization and special solutions to the discrete {S}chwarzian
  {K}d{V} equation.
\newblock {\em J. of Math-for-Industry}, 3:53--62, 2011.

\bibitem{NS}
S.~Yoo-Kong and F.~Nijhoff.
\newblock Elliptic $({N},{N}')$-soliton solutions of the lattice
  {K}adomtsev-{P}etviashvili equation.
\newblock {\em J. Math. Phys.}, 54:043511, 2013.

\bibitem{ABS}
V.~E. Adler, A.~I. Bobenko, and Yu.~B. Suris.
\newblock Classification of integrable equations on quad-graphs. {T}he
  consistency approach.
\newblock {\em Comm. Math. Phys.}, 233(3):513--543, 2003.

\bibitem{ABS2}
V.~E. Adler, A.~I. Bobenko, and Yu.~B. Suris.
\newblock Discrete nonlinear hyperbolic equations. {C}lassification of
  integrable cases.
\newblock {\em Funct. Anal. Appl.}, 43(1):3--17, 2009.

\bibitem{Boll}
R.~Boll.
\newblock Classification of 3{D} consistent quad-equations.
\newblock {\em J. Nonlin. Math. Phys.}, 18(3):337--65, 2011.

\bibitem{AdVeQ}
V.~E. Adler and A.~P. Veselov.
\newblock Cauchy problem for integrable discrete equations on quad-graphs.
\newblock {\em Acta Appl. Math.}, 84(2):237--262, 2004.

\bibitem{WE10}
S.~Mizoguchi and Y.~Yamada.
\newblock {$W(E_{10})$} symmetry, {M}-theory and {P}ainlev\'e equations.
\newblock {\em Physics Letters B}, 537(1):130--140, 2002.

\bibitem{8PC}
T.~Takenawa.
\newblock Discrete dynamical systems associated with the configuration space of
  8 points in ${{P}}^3({{C}})$.
\newblock {\em Communications in Mathematical Physics}, 246(1):19--42, 2004.

\bibitem{tt}
T.~Tsuda.
\newblock A geometric approach to tau-functions of difference {P}ainlev\'e
  equations.
\newblock {\em Letters in Mathematical Physics}, 85(1):65--78, 2008.

\bibitem{Looj}
E.~Looijenga.
\newblock Rational surfaces with an anticanonical cycle.
\newblock {\em Ann. of Math.}, 114(2):267--322, 1981.

\bibitem{sm}
S.~Mukai.
\newblock Geometric realization of {T}-shaped root systems and counterexamples
  to {H}ilbert’s fourteenth problem.
\newblock In V.~L. Popov, editor, {\em Algebraic Transformation Groups and
  Algebraic Varieties}, volume 132 of {\em Encyclopaedia of Mathematical
  Sciences}, pages 123--129. Springer Berlin, Heidelberg, 2004.

\end{thebibliography}
\end{document}